\documentclass[11pt]{article}

\usepackage{amssymb,amsmath}
\usepackage{amsthm}
\usepackage{amsfonts}
\usepackage{latexsym}
\usepackage{geometry}
\usepackage{enumerate}
\usepackage{cite}
\usepackage{graphicx}

\usepackage[latin1]{inputenc}

\newtheorem{thm}{Theorem}
\newtheorem{cor}[thm]{Corollary}

\newtheorem{lem}[thm]{Lemma}
\newtheorem{prop}[thm]{Proposition}

\theoremstyle{definition}
\newtheorem{defn}{Definition}
\newtheorem{exa}{Example}

\newcommand{\conv}{{\textrm{conv} }}

\newcommand{\gp}[2]{\ensuremath{#1^{#2}}}
\newcommand{\gpp}[2]{\ensuremath{#1^{N[#2]}}}
\newcommand{\gme}[2]{\ensuremath{#1^{#2}}}
\newcommand{\Path}[1]{\ensuremath{{\sf P}_{#1}}}
\newcommand{\Polyg}[1]{\ensuremath{{\mathcal P}_{{\cal G}(#1)}}}
\newcommand{\Polyc}[1]{\ensuremath{{\mathcal P}_{{\cal C}(#1)}}}
\newcommand{\Star}[1]{\ensuremath{{\sf S}_{#1}}}
\newcommand{\C}[1]{\ensuremath{{\sf C}_{#1}}}
\newcommand{\K}[1]{\ensuremath{{\sf K}_{#1}}}
\newcommand{\CS}[2]{\ensuremath{{\sf CS}_{#1,#2}}}
\newcommand{\tur}[2]{\ensuremath{{\sf T}_{#1,#2}}}
\newcommand{\turc}[2]{\ensuremath{{\sf TC}_{#1,#2}}}
\newcommand{\Dec}[4]{\ensuremath{{\mathcal D}(#1,#2,#3,#4)}}
\newcommand{\F}{F}
\newcommand{\al}{\alpha}
\newcommand{\fT}{f_{\sf T}}
\newcommand{\fTC}{f_{\sf TC}}
\newcommand{\Lg}[2]{\ensuremath{{\sf L}_{#1} (#2)}}
\newcommand{\Lc}[2]{\ensuremath{{\sf L}'_{#1} (#2)}}

\makeatletter  

\title{Fibonacci Index and Stability Number of Graphs: \\a Polyhedral Study}
\author{V\'eronique Bruy\`ere\footnote{Department of Theoretical Computer Science, Universit\'e de  Mons-Hainaut, Avenue du Champ de Mars 6, B-7000 Mons, Belgium.} \and 
Hadrien M\'elot\footnotemark[1] $^{,}$\footnote{Charg\'e de Recherches F.R.S.-FNRS. Corresponding author. E-mail: {\tt hadrien.melot@umh.ac.be}.}
}

\begin{document}

%






%


\maketitle
\vspace*{0.2cm}

\hrule
\vspace*{0.2cm}
\small
\noindent
\textbf{Abstract.} The Fibonacci index of a graph is the number of its stable sets. This parameter is widely studied and has applications in chemical graph theory. In this paper, we establish tight upper bounds for the Fibonacci index in terms of the stability number and the order of general graphs and connected graphs. Tur\'an graphs frequently appear in extremal graph theory. We show that Tur\'an graphs and a connected variant of them are also extremal for these particular problems. We also make a polyhedral study by establishing all the optimal linear inequalities for the stability number and the Fibonacci index, inside the classes of general and connected graphs of order $n$.

\vspace*{0.2cm}

\vspace*{0.2cm}
\hrule

\normalsize

\section{Introduction}

The Fibonacci index $\F(G)$ of a graph $G$ was introduced in 1982 by Prodinger and Tichy~\cite{PT82} as the number of stable sets in $G$. In 1989, Merrifield and Simmons~\cite{MS89} introduced independently this parameter in the chemistry literature\footnote{The Fibonacci index is called the Fibonacci number by Prodinger and Tichy~\cite{PT82}.  Merrifield and Simmons introduced it as the $\sigma$-index~\cite{MS89}, also known as the Merrifield-Simmons index.}. They showed that there exist correlations between the boiling point and the Fibonacci index of a molecular graph. Since, the Fibonacci index has been widely studied, especially during the last few years. The majority of these recent results appeared in chemical graph theory~\cite{LLW03, LZG05, TW05, Wag07, WH06, WHW07} and in extremal graph theory~\cite{HW07, KTWZ07, PV, PV05, PV06}. 

In this literature, several results are bounds for $\F(G)$ among graphs in particular classes. Lower and upper bounds inside the classes of general graphs, connected graphs, and trees are well known (see Section~\ref{sec_prelim}).  Several authors give a characterization of trees with maximum Fibonacci index inside the class ${\cal T}(n,k)$ of trees with order $n$ and a fixed parameter $k$. For example, Li et al.~\cite{LZG05} determine such trees when $k$ is the diameter; Heuberger and Wagner~\cite{HW07} when $k$ is the maximum degree; and Wang et al.~\cite{WHW07} when $k$ is the number of pending vertices. Unicyclic graphs are also investigated in similar ways~\cite{PV05,PV06,WH06}. 

The Fibonacci index and the stability number of a graph are both related to stable sets. Hence, it is natural to use the stability number as a parameter to determine bounds for $\F(G)$. Let ${\cal G}(n,\al)$ and ${\cal C}(n,\al)$ be the classes of -- respectively general and connected -- graphs with order $n$ and stability number $\al$. The lower bound for the Fibonacci index is known for graphs in these classes. Indeed, Pedersen and Vestergaard~\cite{PV06} give a simple proof to show that if $G \in {\cal G}(n,\al)$ or $G \in {\cal C}(n,\al)$, then $\F(G) \ge 2^\al + n - \al$. Equality occurs if and only if $G$ is a complete split graph (see Section~\ref{sec_prelim}). In this article, we determine upper bounds for $\F(G)$ in the classes ${\cal G}(n,\al)$ and ${\cal C}(n,\al)$. In both cases, the bound is tight for every possible value of $\al$ and $n$ and the extremal graphs are characterized.

A Tur\'an graph is the union of disjoint balanced cliques. Tur\'an graphs frequently appear in extremal graph theory. For example, the well-known Theorem of Tur\'an~\cite{Turan} states that these graphs have minimum size inside ${\cal G}(n,\al)$. We show in Section~\ref{sec_graph} that Tur\'an graphs have also  maximum Fibonacci index inside ${\cal G}(n,\al)$. Observe that removing an edge in a graph strictly increases its Fibonacci index. Indeed, all existing stable sets remain and there is at least one more new stable set: the two vertices incident to the deleted edge. Therefore, we might have the intuition that the upper bound for $\F(G)$ is a simple consequence of the Theorem of Tur\'an. However, we show that it is not true (see Sections~\ref{sec_prelim} and \ref{sec_conc}). The proof uses structural properties of $\al$-critical graphs. 

Graphs in ${\cal C}(n,\al)$ which maximize $\F(G)$ are characterized in Section~\ref{sec_conn}. We call them Tur\'an-connected graphs since they are a connected variant of Tur\'an graphs. It is interesting to note that these graphs again minimize the size inside ${\cal C}(n,\al)$. Hence, our results lead to questions about the relations between the Fibonacci index, the stability number, the size and the order of graphs. These questions are summarized in Section~\ref{sec_conc}.

In Section~\ref{sec_poly}, we further extend our results by a polyhedral study of the relations among the stability number and the Fibonacci index. Indeed, we state all the optimal linear inequalities for the stability number and the Fibonacci index, inside the classes of general and connected graphs of order $n$. 

The major part of the results of this article has been published in Ref.~\cite{GphTuran}.

\section{Basic properties} \label{sec_prelim}

In this section, we suppose that the reader is familiar with usual notions of graph theory (we refer to Berge~\cite{Berge01} for more details). First, we fix our terminology and notation. We then recall the notion of $\al$-critical graphs and give properties of such graphs, used in the next sections. We end with some basic properties of the Fibonacci index of a graph.

\subsection{Notations}

Let $G=(V,E)$ be a simple and undirected graph order $n(G) = | V |$ and size $m(G) = |E|$.  For a vertex $v \in V(G)$, we denote by $N(v)$ the neighborhood of $v$; its closed neighborhood is defined as $\mathcal{N}(v) = N(v) \cup \{v\}$. The degree of a vertex $v$ is denoted by $d(v)$ and the maximum degree of~$G$ by $\Delta(G)$. We use notation $G \simeq H$ when $G$ and $H$ are isomorphic graphs. The complement of $G$ is denoted by $\overline{G}$.

The \emph{stability number} $\al(G)$ of a graph $G$ is the number of vertices of a maximum stable set of $G$. Clearly, $1 \leq \al(G) \leq n(G)$, and $1 \leq \al(G) \leq n(G)-1$ when $G$ is connected.

\begin{defn}
We denote by $\gp{G}{v}$ the induced subgraph obtained by removing a vertex $v$ from a graph $G$. Similarly, the graph $\gpp{G}{v}$ is the induced subgraph obtained by removing the closed neighborhood of $v$. Finally, the graph obtained by removing an edge $e$ from $G$ is denoted by $\gme{G}{e}$.
\end{defn}

Classical graphs of order $n$ are used in this article: the complete graph $\K{n}$, the path $\Path{n}$, the cycle $\C{n}$, the star $\Star{n}$ (composed by one vertex adjacent to $n-1$ vertices of degree 1) and the complete split graph $\CS{n}{\al}$ (composed of a stable set of $\al$ vertices, a clique of $n-\al$ vertices and each vertex of the stable set is adjacent to each vertex of the clique). The complete split graph $\CS{7}{3}$ is depicted in Figure~\ref{fig_exagr}.

We also deeply study the two classes of Tur\'an graphs and Tur\'an-connected graphs. A \emph{Tur\'an graph} $\tur{n}{\al}$ is a graph of order $n$ and a stability number $\al$ such that $1 \leq \al \leq n$, that is defined as follows. It is the union of $\al$ disjoint balanced cliques (that is, such that their orders differ from at most one)~\cite{Turan}. These cliques have thus $\lceil \frac{n}{\al} \rceil$ or $\lfloor \frac{n}{\al} \rfloor$ vertices. We now define a \emph{Tur\'an-connected graph} $\turc{n}{\al}$ with $n$ vertices and a stability number $\al$ where $1 \leq \al \leq n-1$. It is constructed from the Tur\'an graph $\tur{n}{\al}$ with $\al-1$ additional edges. Let $v$ be a vertex of one clique of size $\lceil \frac{n}{\al} \rceil$, the additional edges link $v$ and one vertex of each remaining cliques. Note that, for each of the two classes of graphs defined above, there is only one graph with given values of $n$ and $\al$, up to isomorphism. 

\begin{exa} \label{exa_Turan}
Figure \ref{fig_exagr} shows the Tur\'an graph $\tur{7}{3}$ and the Tur\'an-connected graph $\turc{7}{3}$. When $\al = 1$, we observe that $\tur{n}{1} \simeq \turc{n}{1} \simeq  \CS{n}{1}  \simeq \K{n}$. When $\al = n$, we have $\tur{n}{n} \simeq \CS{n}{n} \simeq  \overline{\K{n}}$, and when $\al = n-1$, we have $\turc{n}{n-1} \simeq \CS{n}{n-1}  \simeq \Star{n}$.
\end{exa}
\begin{figure}[!ht]
\begin{center}
\includegraphics{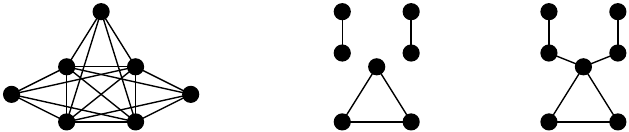}
\caption{The graphs $\CS{7}{3}$, $\tur{7}{3}$ and $\turc{7}{3}$} \label{fig_exagr}
\end{center}
\end{figure}

\subsection{$\al$-critical graphs}

We recall the notion of $\al$-critical graphs~\cite{ErdG61, Joret07, Lovasz86}. An edge $e$ of a graph $G$ is $\al$\emph{-critical} if $\al(\gme{G}{e}) > \al(G)$, otherwise it is called $\al$\emph{-safe}. A graph is said to be $\al$\emph{-critical} if all its edges are $\al$-critical. By convention, a graph with no edge is also $\al$-critical. These graphs play an important role in extremal graph theory~\cite{Joret07}, and also in our proofs.

\begin{exa} \label{exa_critical}
Simple examples of $\al$-critical graphs are complete graphs and odd cycles. Tur\'an graphs are also $\al$-critical. On the contrary, Tur\'an-connected graph are not $\al$-critical, except when $\al =  1$.
\end{exa}

We state some interesting properties of $\al$-critical graphs.

\begin{lem} \label{lem_crit1}
Let $G$ be an $\al$-critical graph. If $G$ is connected, then the graph $\gp{G}{v}$ is connected for all vertices $v$ of $G$.
\end{lem}

\begin{proof}
We use two known results on $\al$-critical graphs (see, e.g., \cite[Chapter 12]{Lovasz86}). If a vertex $v$ of an $\al$-critical graph has degree 1, then $v$ and its neighbor $w$ form a connected component of the graph. Every vertex of degree at least 2 in an $\al$-critical graph is contained in a cycle. 

Hence, by the first result, the minimum degree of $G$ equals 2, except if $G \simeq \K{2}$. Clearly $\gp{G}{v}$ is connected by the second result or when $G \simeq \K{2}$.
\end{proof}

\begin{lem} \label{lem_crit}
Let $G$ be an $\al$-critical graph. Let $v$ be any vertex of $G$ which is not isolated. Then,
\[
 \al(G) = \al(\gp{G}{v}) = \al(\gpp{G}{v}) + 1.
\]
\end{lem}

\begin{proof}
Let $e =vw$ be an edge of $G$ containing $v$. Then, there exist in $G$ two maximum stable sets $S$ and $S'$, such that $S$ contains $v$, but not $w$, and $S'$ contains $w$, but not $v$ (see, e.g., \cite[Chapter 12]{Lovasz86}). Thus, $\al(G) = \al(\gp{G}{v})$ due to the existence of $S'$. The set $S$ avoids each vertex of $N(v)$. Hence, $S \setminus \{v\}$ is a stable set of the graph $\gpp{G}{v}$ of size $\al(G) - 1$. It is easy to check that this stable set is maximum.
\end{proof}

\subsection{Fibonacci index}

Let us now recall the Fibonacci index of a graph~\cite{PT82, MS89}. The \emph{Fibonacci index} $\F(G)$ of  a graph~$G$ is the number of all the stable sets in $G$, including the empty set. The following lemma about $\F(G)$ is well-known (see \cite{GP86, LZG05, PT82}). It is used intensively through the article.

\begin{lem} \label{lem_main1}
Let $G$ be a graph. 
\begin{itemize}
\item Let $e$ be an edge of $G$, then $\F(G) < \F(\gme{G}{e})$.
\item Let $v$ be a vertex of $G$, then $\F(G) = \F(\gp{G}{v}) + \F(\gpp{G}{v})$.
\item
If $G$ is the union of $k$ disjoint graphs $G_i$, $1\leq i\leq k$, then $\F(G) = \prod_{i=1}^k \F(G_i)$.
\end{itemize}
\end{lem}

\begin{exa} \label{exa_F}
We have $\F(\K{n}) = n+1$, $\F(\overline{\K{n}}) = 2^n$, $\F(\Star{n}) = 2^{n-1} + 1$ and $\F(\Path{n}) = f_{n+2}$ (recall that the sequence of Fibonacci numbers $f_n$ is $f_0 = 0, f_1= 1$ and $f_n = f_{n-1} + f_{n-2}$  for $n > 1$).
\end{exa}

Prodinger and Tichy~\cite{PT82} give simple lower and upper bounds for the Fibonacci index. We recall these bounds in the next lemma. 
\begin{lem} \label{lem_main2}
Let $G$ be a graph of order $n$.
\begin{itemize}
\item Then $n+1 \le \F(G) \le 2^n$ with equality if and only if $G \simeq \K{n}$ (lower bound) and $G \simeq \overline{\K{n}}$ (upper bound).
\item If $G$ is connected, then $n+1 \le \F(G) \le 2^{n-1} + 1$ with equality if and only if $G \simeq \K{n}$ (lower bound) and $G \simeq \Star{n}$ (upper bound).
\item If $G$ is a tree, then $f_{n+2} \le \F(G) \le 2^{n-1} + 1$ with equality if and only if $G \simeq \Path{n}$ (lower bound) and $G \simeq \Star{n}$ (upper bound).
\end{itemize}
\end{lem}

We denote by ${\cal G}(n,\al)$ the class of general graphs with order $n$ and stability number $\al$; and by ${\cal C}(n,\al)$ the class of connected graphs with order $n$ and stability number $\al$. Pedersen and Vestergaard~\cite{PV06} characterize graphs with minimum Fibonacci index as indicated in the following theorem.
\begin{thm} \label{thm_lb}
Let $G$ be a graph inside ${\cal G}(n,\al)$ or ${\cal C}(n,\al)$, then 
\[
\F(G) \ge 2^\al + n - \al,
\]
with equality if and only if $G \simeq \CS{n}{\al}$.
\end{thm}
The aim of this article is the study of graphs with maximum Fibonacci index inside the two classes ${\cal G}(n,\al)$ and ${\cal C}(n,\al)$. The system GraPHedron~\cite{GphDesc} allows a formal framework to conjecture optimal relations among a set of graph invariants. Thanks to this system, graphs with maximum Fibonacci index inside each of the two previous classes have been computed for small values of $n$~\cite{GPHRepFiboAlpha}. We observe that these graphs are isomorphic to Tur\'an graphs for the class ${\cal G}(n,\al)$, and to Tur\'an-connected graphs for the class ${\cal C}(n,\al)$. For the class ${\cal C}(n,\al)$, there is one exception when $n = 5$ and $\al = 2$: both the cycle $\C{5}$ and the graph $\turc{5}{2}$ have maximum Fibonacci index. 

Recall that the classical Theorem of Tur\'an~\cite{Turan} states that Tur\'an graphs $\tur{n}{\al}$ have minimum size inside ${\cal G}(n,\al)$. We might think that Tur\'an graphs have maximum Fibonacci index inside ${\cal G}(n,\al)$ as a direct corollary of the Theorem of Tur\'an and Lemma~\ref{lem_main1}. This argument is not correct since removing an $\al$-critical edge increases the stability number. Therefore,  Lemma~\ref{lem_main1} only implies that graphs with maximum Fibonacci index inside ${\cal G}(n,\al)$ are $\al$-critical graphs. In Section~\ref{sec_conc}, we make further observations on the relations between the size and the Fibonacci index inside the classes ${\cal G}(n,\al)$ and ${\cal C}(n,\al)$.

There is another interesting property of Tur\'an graphs related to stable sets. Byskov~\cite{Byskov04} establish that Tur\'an graphs have maximum number of maximal stable sets inside ${\cal G}(n,\al)$. The Fibonacci index counts not only the maximal stable sets but all the stable sets. Hence, the fact that Tur\'an graphs maximize $F(G)$ cannot be simply derived from the result of Byskov.

\section{General graphs}  \label{sec_graph}

In this section, we study graphs with maximum Fibonacci index inside the class ${\cal G}(n,\al)$. These graphs are said to be \emph{extremal}. For fixed values of $n$ and $\al$, we show that there is one extremal graph up to isomorphism, the Tur\'an graph $\tur{n}{\al}$ (see Theorem~\ref{thm_genUB}). 

Before establishing this result, we need some auxiliary results. We denote by $\fT(n,\al)$ the Fibonacci index of the Tur\'an graph $\tur{n}{\al}$. By Lemma~\ref{lem_main1}, its value is equal to
\[
\fT(n,\al) = \left(\left\lceil\frac{n}{\al}\right\rceil + 1\right)^p \  \left(\left\lfloor\frac{n}{\al} \right\rfloor + 1 \right)^{\al-p},
\]
where $p = (n \mod \al)$. We have also the following inductive formula.

\begin{lem} \label{lem_ind}
Let $n$ and $\al$ be integers such that $1 \le \al \le n$. Then
\[
\fT(n,\al) = \left\{ 
\begin{array}{l l l}
n+1 & \textrm{if} & \al = 1,\\
2^n & \textrm{if} & \al = n,\\
\fT(n-1,\al) +  \fT(n-\left\lceil \frac{n}{\al}\right\rceil,\al -1) & \textrm{if} & 2 \le \al \le n-1.\\
\end{array}
\right.
\]
\end{lem}

\begin{proof}
The cases $\al = 1$ and $\al = n$ are trivial (see Example~\ref{exa_F}). Suppose $2 \le \al \le n-1$. Let $v$ be a vertex of $\tur{n}{\al}$ with maximum degree. Thus $v$ is in a $\left\lceil\frac{n}{\al}\right\rceil$-clique. As $\al < n$, the vertex $v$ is not isolated. Therefore $\gp{\tur{n}{\al}}{v} \simeq \tur{n-1}{\al}$. As $\al \geq 2$, the graph $\gpp{\tur{n}{\al}}{v}$ has at least one vertex, and $\gpp{\tur{n}{\al}}{v} \simeq \tur{n-\left\lceil \frac{n}{\al}\right\rceil}{\al -1}$. By Lemma~\ref{lem_main1}, we obtain
\[
\fT(n,\al) = \fT(n-1,\al) +  \fT(n-\left\lceil \frac{n}{\al}\right\rceil,\al -1). \qedhere
\]
\end{proof}

A consequence of Lemma~\ref{lem_ind} is that $\fT(n-1,\al) < \fT(n,\al)$. Indeed, the cases $\al = 1$ and $\al = n$ are trivial, and the term $\fT(n-\left\lceil \frac{n}{\al}\right\rceil,\al -1)$ is always strictly positive when $2 \le \al \le n-1$. 

\begin{cor} \label{cor_f1inc}
The function $\fT(n,\al)$ is strictly increasing in $n$ when $\al$ is fixed.
\end{cor}

We now state the upper bound on the Fibonacci index of graphs in the class ${\cal G}(n,\al)$.

\begin{thm} \label{thm_genUB}
Let $G$ be a graph of order $n$ with a stability number $\al$, then
\[
\F(G) \le \fT(n,\al),
\]
with equality if and only if $G \simeq \tur{n}{\al}$.
\end{thm}

\begin{proof}
The cases $\al = 1$ and $\al = n$ are straightforward. Indeed $G \simeq \tur{n}{1}$ when $\al = 1$, and $G \simeq \tur{n}{n}$ when $\al = n$. We can assume that $2 \le \al \le n-1$, and thus $n \geq 3$. We now prove by induction on $n$ that if $G$ is extremal, then it is isomorphic to $\tur{n}{\al}$. 

The graph $G$ is $\al$-critical. Otherwise, there exists an edge $e \in E(G)$ such that $\al(G) = \al(\gme{G}{e})$, and by Lemma \ref{lem_main1}, $F(G) < F(\gme{G}{e})$. This is a contradiction with $G$ being extremal.

Let us compute $F(G)$ thanks to Lemma~\ref{lem_main1}. Let $v \in V(G)$ of maximum degree $\Delta$. The vertex $v$ is not isolated since $\al < n$. Thus by Lemma~\ref{lem_crit}, $\al(\gp{G}{v})  =  \al$ and $\al(\gpp{G}{v}) = \al - 1$. On the other hand, If $\chi$ is the chromatic number of $G$, it is well-known that $n \le \chi \ . \ \al$ (see, e.g., Berge~\cite{Berge01}), and that $\chi \le \Delta + 1$ (see Brooks~\cite{Brooks41}). It follows that 
\begin{equation} \label{eq_ngpp} 
n(\gpp{G}{v}) = n - \Delta - 1 \leq n - \left\lceil \frac{n}{\al}\right\rceil.
\end{equation}
Note that $n(\gpp{G}{v}) \geq 1$ since $\al \geq 2$. 

We can apply the induction hypothesis on the graphs $\gp{G}{v}$ and $\gpp{G}{v}$. We obtain

\[
\begin{array}{ r c l p{0.35 \textwidth}}
\fT(n,\al) & \le & \F(G) &\textrm{as $G$ is extremal},\\
& = & \F(\gp{G}{v}) + \F(\gpp{G}{v}) & \textrm{by Lemma~\ref{lem_main1},}\\ 
& \le & \fT(n(\gp{G}{v}),\al(\gp{G}{v})) + \fT(n(\gpp{G}{v}),\al(\gpp{G}{v})) & \textrm{by induction,}\\
& = & \fT(n-1,\al) + \fT(n-\Delta-1,\al-1) &\\
& \le &  \fT(n-1,\al) +  \fT(n-\left\lceil \frac{n}{\al}\right\rceil,\al -1) & \textrm{by Eq.~\eqref{eq_ngpp} and Corollary~\ref{cor_f1inc},}\\
& = & \fT(n,\al) & \textrm{by Lemma~\ref{lem_ind}}.
\end{array}
\]
Hence  equality holds everywhere. In particular, by induction, the graphs $\gp{G}{v}$, $\gpp{G}{v}$ are extremal, and $\gp{G}{v} \simeq \tur{n-1}{\al}$, $\gpp{G}{v} \simeq \tur{n-\left\lceil \frac{n}{\al}\right\rceil}{\al -1}$. Coming back to $G$ from $\gp{G}{v}$ and $\gpp{G}{v}$ and recalling that $v$ has maximum degree, it follows that $G \simeq \tur{n}{\al}$.
\end{proof}

Corollary~\ref{cor_f1inc} states that $\fT(n,\al)$ is increasing in $n$. It was an easy consequence of Lemma~\ref{lem_ind}. The function $\fT(n,\al)$ is also increasing in $\al$. Theorem~\ref{thm_genUB} can be used to prove this fact easily as shown now.

\begin{cor} \label{cor_f1inc2}
The function $\fT(n,\al)$ is strictly increasing in $\al$ when $n$ is fixed.
\end{cor}

\begin{proof}
Suppose $2 \le \al \le n-1$. By Lemma~\ref{lem_main2} it is clear that $\fT(n,1) < \fT(n,\al) < \fT(n,n)$. Now, let $e$ be an edge of $\tur{n}{\al}$. Clearly $\al(\gme{\tur{n}{\al}}{e}) = \al+1$. Moreover, by Lemma~\ref{lem_main1} and Theorem~\ref{thm_genUB},
\[
\F(\tur{n}{\al}) < \F(\gme{\tur{n}{\al}}{e}) < \F(\tur{n}{\al+1}).
\]
Therefore, $\fT(n,\al) < \fT(n,\al+1)$.
\end{proof}

\section{Connected graphs}  \label{sec_conn}

We now consider graphs with maximum Fibonacci index inside the class ${\cal C}(n,\al)$. Such graphs are called \emph{extremal}. If $G$ is connected, the bound of Theorem \ref{thm_genUB} is clearly not tight, except when $\al = 1$, that is, when $G$ is a complete graph. We are going to prove that there is one extremal graph up to isomorphism, the Tur\'an-connected graph $\turc{n}{\al}$, with the exception of the cycle $\C{5}$ (see Theorem~\ref{thm_conUB}).  First, we need preliminary results and definitions to prove this theorem.

We denote by $\fTC(n,\al)$ the Fibonacci index of the Tur\'an-connected graph $\turc{n}{\al}$. An inductive formula for its value is given in the next lemma.

\begin{lem} \label{lem_f2}
Let $n$ and $\al$ be integers such that $1 \le \al \le n-1$. Then
\[
\fTC(n,\al) = \left\{ 
\begin{array}{l l l}
n+1 & \textrm{if} & \al = 1,\\
2^{n-1} + 1 & \textrm{if} & \al = n-1,\\
\fT(n-1,\al) + \fT(n',\al') & \textrm{if} & 2 \le \al \le n-2,\\
\end{array}
\right.
\]
where $n' = n-\left\lceil\frac{n}{\al}\right\rceil-\al+1$ and $\al' = \min(n',\al - 1)$.
\end{lem}
\begin{proof}
The cases $\al = 1$ and $\al = n-1$ are trivial by Lemma~\ref{lem_main2}. Suppose now that $2 \le \al \le n-2$. Let $v$ be a vertex of maximum degree in $\turc{n}{\al}$. We apply Lemma~\ref{lem_main1} to compute $F(\turc{n}{\al})$. Observe that the graphs $\gp{\turc{n}{\al}}{v}$ and $\gpp{\turc{n}{\al}}{v}$ are both Tur\'an graphs when $2 \le \al \le n-2$. 

The graph $\gp{\turc{n}{\al}}{v}$ is isomorphic to $\tur{n-1}{\al}$. Let us show that  $\gpp{\turc{n}{\al}}{v}$ is isomorphic to $\tur{n'}{\al'}$. By definition of a Tur\'an-connected graph, $d(v)$ is equal to $\left\lceil\frac{n}{\al}\right\rceil+\al-2$. Thus 
\[
n(\gpp{\turc{n}{\al}}{v}) = n - d(v) - 1 = n'.
\]
If $\al < \frac{n}{2}$, then $\turc{n}{\al}$ has a clique of order at least 3 and $\al(\gpp{\turc{n}{\al}}{v}) = \al - 1 \leq n'$. Otherwise, $\gpp{\turc{n}{\al}}{v} \simeq \overline{\K{n'}}$ and $\al(\gpp{\turc{n}{\al}}{v}) = n' \leq \al - 1$. Therefore $\al(\gpp{\turc{n}{\al}}{v}) =  \min(n',\al - 1)$ in both cases.

By Lemma~\ref{lem_main1}, these observations leads to 
\[
\fTC(n,\al) =  \fT(n-1,\al) + \fT(n',\al'). \qedhere
\]
\end{proof}

\begin{defn}
A \emph{bridge} in a connected graph $G$ is an edge $e \in E(G)$ such that the graph $\gme{G}{e}$ is no more connected. To a bridge $e = v_1 v_2$ of $G$ which is $\al$-safe, we associate a \emph{decomposition} $\Dec{G_1}{v_1}{G_2}{v_2}$ such that $v_1 \in V(G_1)$, $v_2 \in V(G_2)$, and $G_1, G_2$ are the two connected components of $\gme{G}{e}$. A decomposition is said to be $\al$-\emph{critical} if $G_1$ is $\al$-critical.
\end{defn}

\begin{lem} \label{lem_compCrit}
Let $G$ be a connected graph. If $G$ is extremal, then either $G$ is $\al$-critical or $G$ has an $\al$-critical decomposition.
\end{lem}
\begin{proof}
We suppose that $G$ is not $\al$-critical and we show that it must contain an $\al$-critical decomposition. 

Let $e$ be an $\al$-safe edge of $G$. Then $e$ must be a bridge. Otherwise, the graph $\gme{G}{e}$ is connected, has the same order and stability number as $G$ and satisfies $\F(\gme{G}{e}) > \F(G)$ by Lemma~\ref{lem_main1}. This is a contradiction with $G$ being extremal. Therefore $G$ contains at least one $\al$-safe bridge defining a decomposition of $G$. 

Let us choose a decomposition $\Dec{G_1}{v_1}{G_2}{v_2}$ such that $G_1$ is of minimum order. Then, $G_1$ is $\al$-critical. Otherwise, $G_1$ contains an $\al$-safe bridge $e' = w_1 w_2$, since the edges of $G$ are $\al$-critical or $\al$-safe bridges by the first part of the proof. Let $\Dec{H_1}{w_1}{H_2}{w_2}$ be the decomposition of $G$ defined by $e'$, such that $v_1 \in V(H_2)$. Then $n(H_1) < n(G_1)$, which is a contradiction. Hence the decomposition $\Dec{G_1}{v_1}{G_2}{v_2}$ is $\al$-critical.
\end{proof}

\begin{thm} \label{thm_conUB}
Let $G$ be a connected graph of order $n$ with a stability number $\al$, then
\[
\F(G) \le \fTC(n,\al),
\]
with equality if and only if $G \simeq \turc{n}{\al}$ when $(n,\al) \neq (5,2)$, and $G \simeq \turc{5}{2}$ or $G \simeq \C{5}$ when $(n,\al) = (5,2)$.
\end{thm}

\begin{proof}
 We prove by induction on $n$ that if $G$ is extremal, then it is isomorphic to $\turc{n}{\al}$ or $\C{5}$. To handle more easily the general case of the induction (in a way to avoid the extremal graph $\C{5}$), we consider all connected graphs with up to 6 vertices as the basis of the induction. For these basic cases, we refer to the report of an exhaustive automated verification~\cite{GPHRepFiboAlpha}. We thus suppose that $n \ge 7$.

We know by Lemma~\ref{lem_compCrit} that either $G$ has an $\al$-critical decomposition or $G$ is $\al$-critical. We consider now these two situations.

\paragraph{1) $G$ has an $\al$-critical decomposition.} We prove in three steps that $G \simeq \turc{n}{\al}$: ($i$) We establish that for every decomposition $\Dec{G_1}{v_1}{G_2}{v_2}$, the graph $G_i$ is extremal and is isomorphic to a Tur\'an-connected graph such that $d(v_i) = \Delta(G_i)$, for $i = 1,2$. ($ii$) We show that if such a decomposition is $\al$-critical, then $G_1$ is a clique. ($iii$) We prove that $G$ is itself isomorphic to a Tur\'an-connected graph.

($i$) For the first step, let $\Dec{G_1}{v_1}{G_2}{v_2}$ be a decomposition of $G$, $n_1$ be the order of $G_1$, and $\al_1$ its stability number. We prove that $G_1 \simeq \turc{n_1}{\al_1}$ such that $d(v_1) = \Delta(G_1)$. The argument is identical for $G_2$. By Lemma~\ref{lem_main1}, we have
\[
 \F(G) = \F(G_1) \F(\gp{G_2}{v_2}) + \F(\gp{G_1}{v_1}) \F(\gpp{G_2}{v_2}).
\]
By the induction hypothesis, $\F(G_1) \leq \fTC(n_1,\al_1)$. The graph $\gp{G_1}{v_1}$ has an order $n_1 - 1$ and a stability number $\leq \al_1$. Hence by Theorem~\ref{thm_genUB} and Corollary~\ref{cor_f1inc2}, $ \F(\gp{G_1}{v_1}) \leq \fT(n_1-1,\al_1)$. It follows that
\begin{equation} \label{eq_con}
 \F(G)  \le \fTC(n_1,\al_1) \F(\gp{G_2}{v_2}) + \fT(n_1 - 1,\al_1) \F(\gpp{G_2}{v_2}).
\end{equation}
As $G$ is supposed to be extremal, equality occurs. It means that $\gp{G_1}{v_1} \simeq \tur{n_1 - 1}{\al_1}$ and $G_1$ is extremal. If $G_1$ is isomorphic to $\C{5}$, then $n_1 = 5$, $\al_1=2$ and $\F(G_1) = \fTC(5,2)$.  However, $\F(\gp{G_1}{v_1}) =  \F(\Path{4}) < \fT(4,2)$. By~\eqref{eq_con}, this leads to a contradiction with $G$ being extremal. Thus, $G_1$ must be isomorphic to $\turc{n_1}{\al_1}$. Moreover, $v_1$ is a vertex of maximum degree of $G_1$. Otherwise, $\gp{G_1}{v_1}$ cannot be isomorphic to the graph $\tur{n_1 - 1}{\al_1}$.

($ii$) The second step is easy. Let $\Dec{G_1}{v_1}{G_2}{v_2}$ be an $\al$-critical decomposition of $G$, that is, $G_1$ is $\al$-critical. By ($i$), $G_1$ is isomorphic to a Tur\'an-connected graph. The complete graph is the only Tur\'an-connected graph which is $\al$-critical. Therefore, $G_1$ is a clique. 

($iii$) We now suppose that $G$ has an $\al$-critical decomposition $\Dec{G_1}{v_1}{G_2}{v_2}$ and we show that $G \simeq \turc{n}{\al}$. Let $n_1$ be the order of $G_1$ and $\al_1$ its stability number. As $v_1 v_2$ is an $\al$-safe bridge, it is clear that $n(G_2) = n - n_1$ and $\al(G_2) = \al - \al_1$. By ($i$) and ($ii$), $G_1$ is a clique (and thus $\al_1 = 1$), $G_2 \simeq\turc{n - n_1}{\al - 1}$, and $v_2$ is a vertex of maximum degree in $G_2$.  

If $\al = 2$, then $G_2$ is also a clique in $G$. By Lemma~\ref{lem_main1} and the fact that $\F(\K{n}) = n+1$ we have,
\[
\begin{array}{r c l}
\F(G) & = & \F(\gp{G}{v_1}) + \F(\gpp{G}{v_1})\\
 & = & n_1 (n - n_1 + 1) + (n-n_1) =  n + n \ n_1 - n_1^2.
\end{array}
\]
When $n$ is fixed, this function is maximized when $n_1 = \frac{n}{2}$. That is, when $G_1$ and $G_2$ are balanced cliques. This appears if and only if $G \simeq \turc{n}{2}$. 

Thus we suppose that $\al \geq 3$. In other words, $G$ contains at least three cliques: the clique $G_1$ of order $n_1$; the clique $H$ containing $v_2$ and a clique $H'$ in $G_2$ linked to $H$ by an $\al$-safe bridge $v_2 v_3$. Let $k =  \frac{n-n_1}{\al-1}$, then the order of $H$ is $\left\lceil k \right\rceil$ and the order of $H'$ is $\left\lceil k \right\rceil$ or $\left\lfloor k \right\rfloor$ (recall that $G_2 \simeq \turc{n-n_1}{\al -1}$). These cliques are represented in Figure~\ref{fig_case1}. 

\begin{figure}[!ht]
\begin{center}
\includegraphics{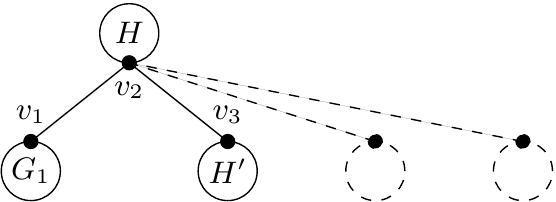}
\caption{Cliques in the graph $G$} \label{fig_case1}
\end{center}
\end{figure}
To prove that $G$ is isomorphic to a Tur\'an-connected graph, it remains to show that the clique $G_1$ is balanced with the cliques $H$ and $H'$. We consider the decomposition defined by the $\al$-safe bridge $v_2 v_3$. By~($i$), $G_1$ and $H$ are cliques of a Tur\'an-connected graph, and $H$ is a clique with maximum order in this graph (recall that $v_2$ is a vertex of maximum degree in $G_2$). Therefore $\left\lceil k \right\rceil -1 \le n_1 \le  \left\lceil k \right\rceil$, showing that $G_1$ is balanced with $H$ and $H'$.

\paragraph{2) $G$ is $\al$-critical.}  Under this hypothesis, we prove that $G$ is a complete graph, and thus is isomorphic to a Tur\'an-connected graph.

Suppose that $G$ is not complete. Let $v$ be a vertex of $G$ with a maximum degree $d(v) = \Delta$. As $G$ is connected and $\al$-critical, the graph $\gp{G}{v}$ is connected by Lemma~\ref{lem_crit1}. By Lemma~\ref{lem_crit}, $\al(\gp{G}{v}) = \al$ and $\al(\gpp{G}{v}) = \al - 1$. Moreover, $n(\gp{G}{v}) = n-1$ and $n(\gpp{G}{v}) = n - \Delta - 1$. By the induction hypothesis and Theorem~\ref{thm_genUB}, we get
\[
 \F(G) = \F(\gp{G}{v}) + \F(\gpp{G}{v}) \le \fTC(n-1,\al) + \fT(n-\Delta-1,\al-1).
\] 
Therefore, $G$ is extremal if and only if $\gpp{G}{v} \simeq \tur{n-\Delta-1}{\al-1}$ and $\gp{G}{v}$ is extremal. However, $\gp{G}{v}$ is not isomorphic to $\C{5}$ as $n \ge 7$. Thus $\gp{G}{v} \simeq \turc{n-1}{\al}$.

So, the graph $G$ is composed by the graph $\gp{G}{v} \simeq \turc{n-1}{\al}$ and an additional vertex $v$ connected to $\turc{n-1}{\al}$ by $\Delta$ edges. 

There must be an edge between $v$ and a vertex $v'$ of maximum degree in $\gp{G}{v}$, otherwise $\gpp{G}{v}$ is not isomorphic to a Tur\'an graph. The vertex $v'$ is adjacent to $\left\lceil \frac{n-1}{\al} \right\rceil + \al - 2$ vertices in $\gp{G}{v}$ and it is adjacent to $v$, that is,
\[
d(v') = \left\lceil \frac{n-1}{\al} \right\rceil + \al - 1.
\]
It follows that 
\begin{equation} \label{eq_del}
\Delta \ge d(v') >  \left\lceil \frac{n-1}{\al} \right\rceil
\end{equation}
as $G$ is not a complete graph. 

On the other hand, $v$ is adjacent to each vertex of some clique $H$ of $\gp{G}{v}$ since $\gpp{G}{v}$ has a stability number $\al - 1$. As this clique has order at most $\left\lceil \frac{n-1}{\al} \right\rceil$, $v$ must be adjacent to a vertex $w \notin H$ by \eqref{eq_del}. 

We observe that the edge $vw$ is $\al$-safe. This is impossible as $G$ is $\al$-critical. It follows that $G$ is a complete graph and the proof is completed. \qedhere
\end{proof}

The study of the maximum Fibonacci index inside the class ${\cal T}(n,\al)$ of trees with order $n$ and stability number $\al$ is strongly related to the study done in this section for the class ${\cal C}(n,\al)$. Indeed, due to the fact that trees are bipartite, a tree in ${\cal T}(n,\al)$ has always a stability number $\al \geq \frac{n}{2}$. Moreover, the Tur\'an-connected graph $\turc{n}{\al}$ is a tree when $\al \ge \frac{n}{2}$. Therefore, the upper bound on the Fibonacci index for connected graphs is also valid for trees. We thus get the next corollary with in addition the exact value of $\fTC(n,\al)$.

\begin{cor} \label{cor_trees} Let $G$ be a tree of order $n$ with a stability number $\al$, then
\[
\F(G) \le 3^{n-\al-1} 2^{2 \al - n + 1}+2^{n-\al-1},
\]
with equality if and only if $G \simeq \turc{n}{\al}$.
\end{cor}

\begin{proof}
It remains to compute the exact value of $\fTC(n,\al)$. When $\al \ge \frac{n}{2}$, the graph $\turc{n}{\al}$ is composed by one central vertex $v$ of degree $\al$ and $\al$ pending paths of length 1 or 2 attached to $v$. An extremity of a pending path of length 2 is a vertex $w$ such that $w \notin  \mathcal{N}(v)$. Thus there are $x = n-\al-1$ pending paths of length 2 since $\mathcal{N}(v)$ has size $\al +1$, and there are $y = \al - x = 2 \al - n + 1$ pending paths of length 1. We apply Lemma~\ref{lem_main1} on $v$ to get
\[
\fTC(n,\al) = \F(\K{2})^x \F(\K{1})^y + \F(\K{1})^x = 3^x 2^y+2^x. \qedhere
\] 
\end{proof}

We conclude this section by showing that the function $\fTC(n,\al)$ is strictly increasing in $n$ and $\al$, as already stated for the function  $\fT(n,\al)$  (see Corollaries \ref{cor_f1inc} and \ref{cor_f1inc2}).

\begin{prop} \label{cor_f2inc}
The function $\fTC(n,\al)$ is strictly increasing in $n$ and $\al$.
\end{prop}

\begin{proof}

We first prove that $\fTC(n,\al)$ is strictly increasing in $n$ when $\al$ is fixed. The cases $\al = 1$ and $\al = n-1$ are obvious by Lemma~\ref{lem_f2} and we suppose that $2 \le \al \le n-2$. Let $n' = n-\left\lceil\frac{n}{\al}\right\rceil-\al+1$ and $\al' = \min(n',\al - 1)$. Also, we note $n'' = n+1-\left\lceil\frac{n+1}{\al}\right\rceil-\al+1$ and $\al'' = \min(n'',\al - 1)$. Observe that $n' \le n''$ and $\al' \le \al''$. We have
$$
\begin{array}{l l l l}
\fTC(n,\al) & = & \fT(n-1,\al) + \fT(n',\al') & \textrm{by Lemma~\ref{lem_f2},}\\
& < & \fT(n,\al) + \fT(n'',\al'') & \textrm{by Corollaries~\ref{cor_f1inc} and~\ref{cor_f1inc2},}\\
& = & \fTC(n+1,\al) & \textrm{by Lemma~\ref{lem_f2}.}\\
\end{array}
$$
Therefore, $\fTC(n,\al) < \fTC(n+1,\al)$. 

We now prove that $\fTC(n,\al)$ is strictly increasing in $\al$ when $n$ is fixed. Let $2 \le \al \le n-2$. Obviously,  $\fTC(n,1) < \fTC(n,\al) < \fTC(n,n-1)$ by Lemma~\ref{lem_main2}. We consider two cases.
\begin{enumerate}[a) ]
\item If $\al < \frac{n}{2}$, then $\turc{n}{\al}$ contains at least one clique $H$ of size at least 3 and the remaining cliques are of size at least 2. Suppose that $G$ is the graph obtained from $\turc{n}{\al}$ by removing an edge inside $H$. Then, $G$ is connected and $\al(G)=\al+1$. Moreover, Lemma~\ref{lem_main1} and Theorem~\ref{thm_conUB} ensure that $\fTC(n,\al) < F(G) < \fTC(n,\al+1)$ and the result follows.
\item Suppose now that $\al \ge \frac{n}{2}$. In this case, $\turc{n}{\al}$ and $\turc{n}{\al+1}$ are trees. Let $x = n-\al-1$, $x'=n-\al-2$, $y = 2\al -n+1$, and $y'=2\al-n+3$. Then, 
$$
\begin{array}{l l l l}
\fTC(n,\al+1) - \fTC(n,\al) & = & 3^{x'} 2^{y'} + 2^{x'} - 3^x 2^y  - 2^x & \textrm{by Corollary~\ref{cor_trees},}\\ 
& = & 3^{x-1} 2^y - 2^{x-1}.
\end{array}
$$

As $\al \le n-2$, we have that $x-1 \ge 0$. Thus, $2^{x-1} \le 3^{x-1}$. Morevover, as $\al \ge \frac{n}{2}$, we have that $y \ge 0$ and thus $2^y \geq 1$. It follows that $3^{x-1} 2^y - 2^{x-1} \geq 0$. The case of equality with $0$ happens when both $x-1 = 0$ and $y = 0$, that is, when $\al = 1$. This never holds since $\al \geq 2$. Therefore $\fTC(n,\al)$ is strictly increasing in $\al$. \qedhere
\end{enumerate}
\end{proof}

\section{Polyhedral study} \label{sec_poly}

In the previous sections, we have stated that the graphs with maximum Fibonacci index inside the classes ${\cal G}(n,\al)$ and ${\cal C}(n,\al)$ are isomorphic to Tur\'an graphs and Tur\'an-connected graphs respectively (see Theorems \ref{thm_genUB} and  \ref{thm_conUB}). These results have been suggested thanks to the system GraPHedron \cite{GPHRepFiboAlpha}.

In this section, we further push the use of the system GraPHedron as outlined in \cite{GphDesc}. Indeed, this framework allows to suggest the set of all optimal linear inequalities among the stability number and the Fibonacci index for graphs inside the class ${\cal G}(n)$ of general graphs of order $n$ and the class ${\cal C}(n)$ of connected graphs of order $n$. That is, it allows to determine for small values of $n$ the complete description of the polytopes
\begin{equation}\label{eq:poly1}
\Polyg{n} = \conv \left\{ (x,y) \ | \ \exists G \in {\cal G}(n), \alpha(G) = x, \F(G) = y \right\},
\end{equation}
\begin{equation}\label{eq:poly2}
\Polyc{n} = \conv \left\{ (x,y) \ | \ \exists G \in {\cal C}(n), \alpha(G) = x, \F(G) = y \right\},
\end{equation}
where \emph{conv} denotes the convex hull.

\begin{figure}[!ht]
\begin{center}
\includegraphics[width=0.45\textwidth]{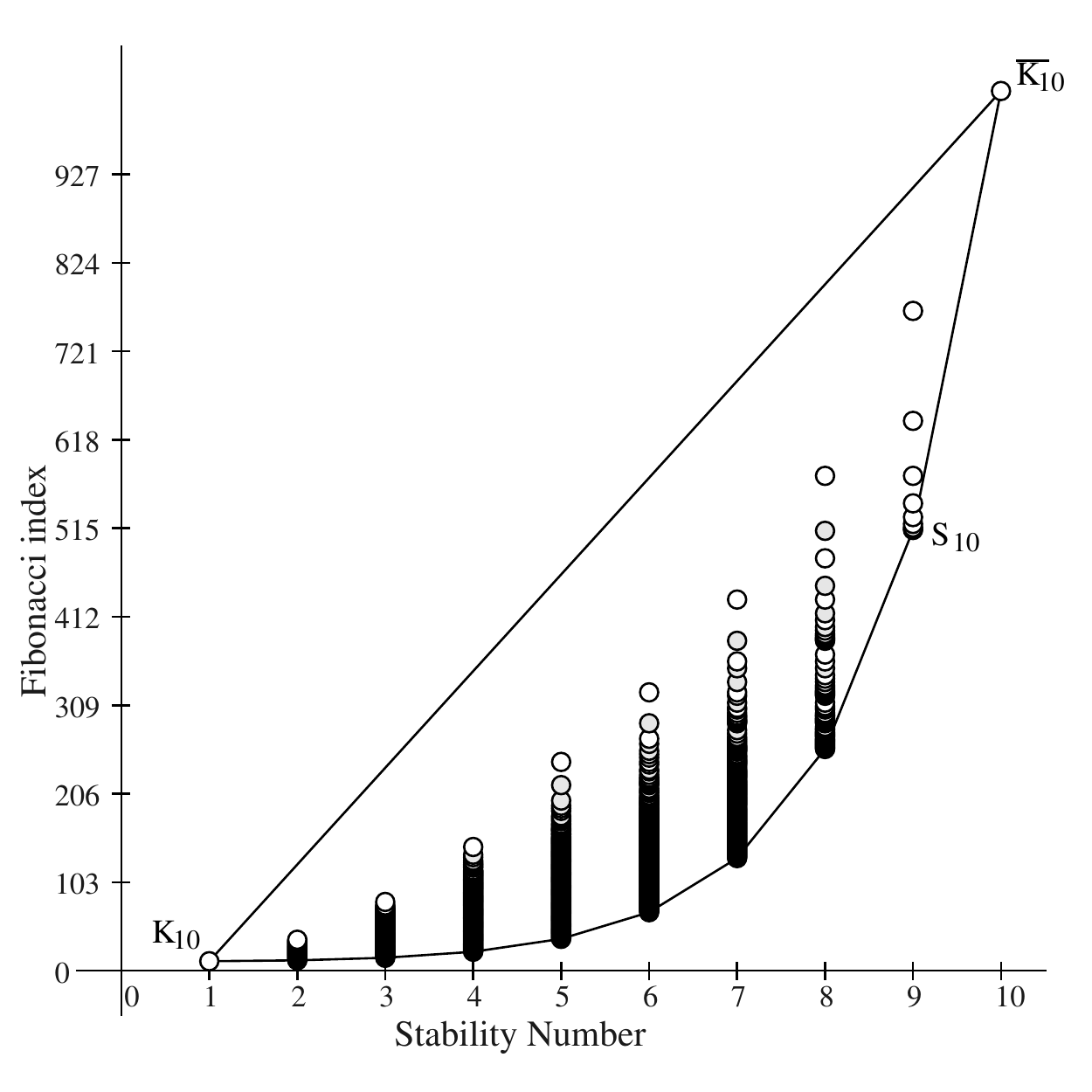} \includegraphics[width=0.45\textwidth]{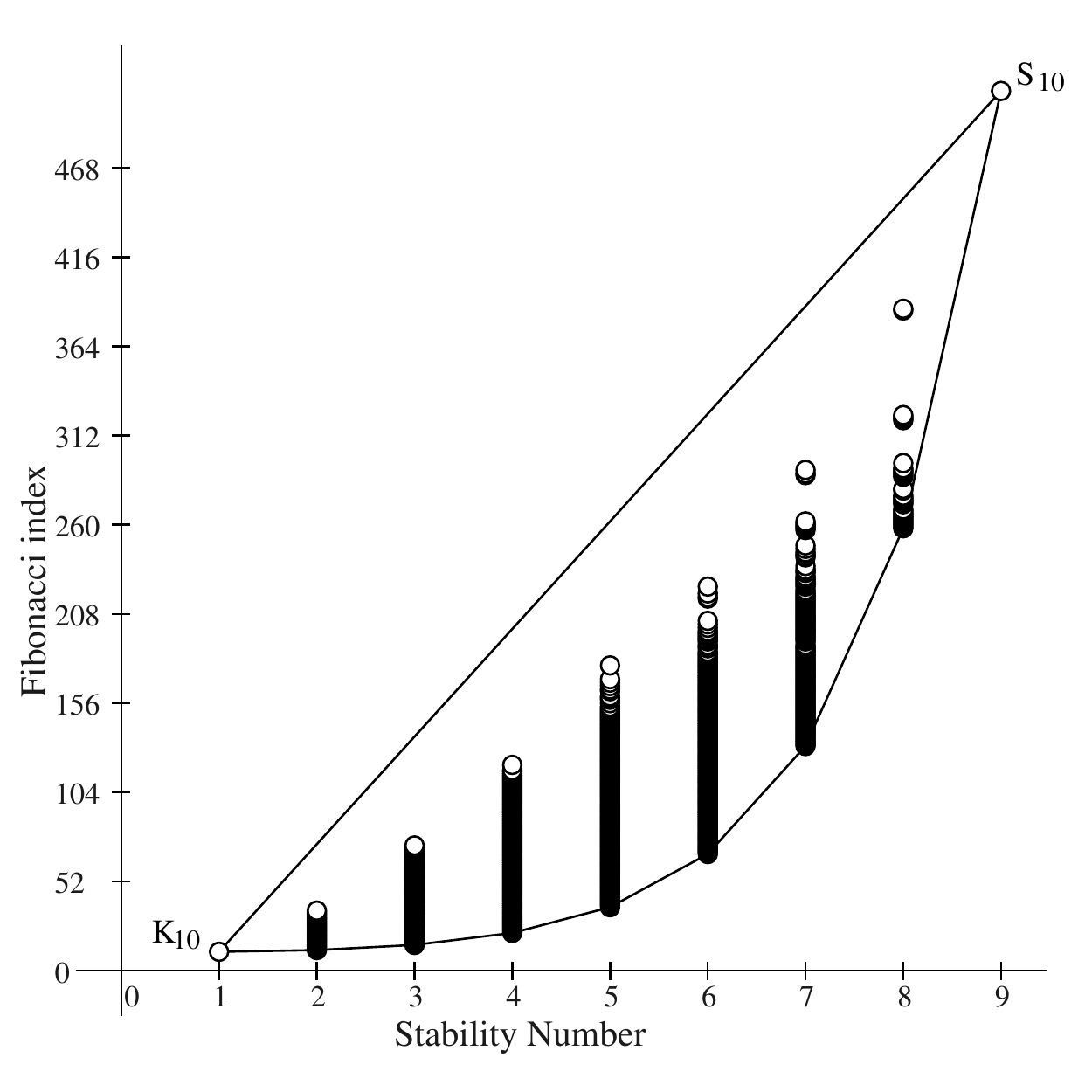}
\caption{The polytopes $\Polyg{10}$ (left) and $\Polyc{10}$ (right)}  \label{fig:polys}
\end{center}
\end{figure}

For example, Figure~\ref{fig:polys} shows the polytopes $\Polyg{n}$ and $\Polyc{n}$ when $n=10$, as given in the reports created by GraPHedron \cite{GPHRepFiboAlpha}. In these representations, we associate to a point $(x,y)$ the set of all graphs with a stability number $x$ and a Fibonacci index $y$, and we say that the point $(x,y)$ \emph{corresponds} to these graphs. For instance, in Figure~\ref{fig:polys}, the point $(1,11)$ corresponds to the graph $\K{10}$, whereas the point $(9,2^9 + 1)$ corresponds to the graph $\Star{10}$.

In this section, we make a polyhedral study in a way to give a complete description of the polytopes $\Polyg{n}$ and $\Polyc{n}$ for all (sufficiently large) values of $n$. More precisely, we are going to describe the \emph{facet defining inequalities} of both polytopes $\Polyg{n}$ and $\Polyc{n}$, that is, their minimal system of linear inequalities. Let us fix some notation:
$$
\Lg{n}{x} = \frac{2^{n} - n -1}{n-1} (x - 1) + n+1, 
$$
$$ \Lc{n}{x} = \frac{2^{n-1} - n }{n-2} (x - 1) + n+1.
$$
The following Theorems~\ref{thm_polyg} and \ref{thm_polyc} give the complete description of $\Polyg{n}$ and $\Polyc{n}$. These theorems will be proved at the end of this section, after some preliminary results.

\begin{thm} \label{thm_polyg}
Let $n \ge 5$. Then the polytope $\Polyg{n}$ has $n$ facets defined by the inequalities 
\begin{eqnarray}
y & \ge & \left(2^k -1 \right) x + 2^k (1-k) + n, \quad \textrm{ for } k = 1, 2, \dots, n-1,\\ \label{eq:otherg}
y & \le & \Lg{n}{x}.
\end{eqnarray}
\end{thm}

\begin{thm} \label{thm_polyc}
Let  $n \ge 8$. Then the polytope $\Polyc{n}$ has $n-1$ facets defined by the inequalities
\begin{eqnarray}
y & \ge & \left(2^k -1 \right) x + 2^k (1-k) + n, \quad \textrm{ for } k = 1, 2, \dots, n-2,\\ \label{eq:otherc}
y & \le &  \Lc{n}{x}.
\end{eqnarray}
\end{thm}

\begin{figure}[!ht]
\begin{center}
\includegraphics[scale=0.7]{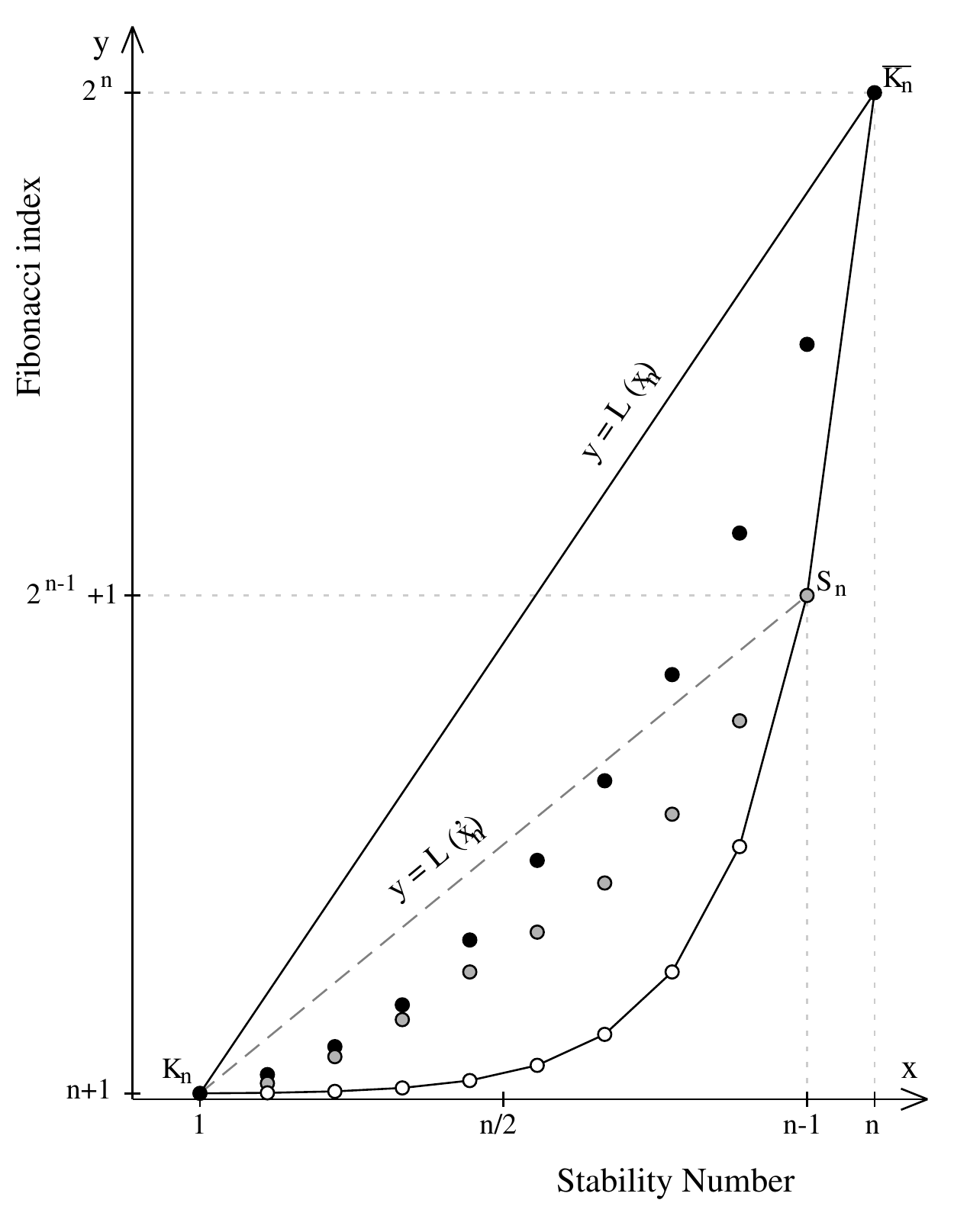}
\caption{Representation of $\Polyg{n}$ and $\Polyc{n}$ together}  \label{fig:poly_all}
\end{center}
\end{figure}

We first make some comments. In Figure~\ref{fig:poly_all}, the two polytopes $\Polyg{n}$ and $\Polyc{n}$ are drawn together.  This gives a graphical summary of the main results stated in Theorems~\ref{thm_lb}, \ref{thm_genUB}, \ref{thm_conUB}, \ref{thm_polyg} and \ref{thm_polyc}:
\begin{itemize}
\item black points correspond to Tur\'an graphs and have maximum $y$-value among general graphs by Theorem~\ref{thm_genUB};
\item grey points correspond to Tur\'an-connected graphs and have maximum $y$-value among connected graphs by Theorem~\ref{thm_conUB};
\item white points correspond to complete split graphs and have minimum $y$-value among general and connected graphs by  Theorem~\ref{thm_lb};
\item the $n$ facets of  $\Polyg{n}$ are the $n-1$ lines joining two consecutive points corresponding to complete split graphs, and the line $y = \Lg{n}{x}$ joining the two points corresponding to $\K{n}$ and $\overline{\K{n}}$ (see Theorem \ref{thm_polyg});
\item the $n-1$ facets of  $\Polyc{n}$ are the $n-2$ lines joining two consecutive points corresponding to (connected) complete split graphs, and the line $y = \Lc{n}{x}$ joining the two points corresponding to $\K{n}$ and $\Star{n}$ (see Theorem \ref{thm_polyc}).
\end{itemize}

In the next lemma, we establish the inequalities (\ref{eq:otherg}) and (\ref{eq:otherc}).

\begin{lem} \label{lem:csineq}
The inequality
\begin{equation} \label{ineq:cs}
y \ge \left(2^k -1 \right) x + 2^k (1-k) + n,
\end{equation}
defines a facet of $\Polyg{n}$ for $k = 1, 2, \dots, n-1$, and a facet of $\Polyc{n}$ for $k = 1, 2, \dots, n-2$.
\end{lem}

\begin{proof}
We know by Theorem~\ref{thm_lb} that the points which have minimum $y$-values are those corresponding to complete split graphs. These points are
$$
\left( x, \ 2^{x}+n-x \right),
$$
which are convexly independent as the function $2^x + n - x$ is strictly convex in $x$. Therefore these points are vertices of  $\Polyg{n}$ and $\Polyc{n}$, for $x = 1, 2, \dots, n-1$, and $x = 1, 2, \dots, n-2$, respectively. Moreover, there can be no other polytope vertices between two consecutive points because $x$ is increasing by step of 1, and there exists a complete split graph for each possible value of $x$.

Ineq.~(\ref{ineq:cs}) can then be derived by computing the equation of the line passing by two consecutive points  $\left( k , \ 2^{k}+n-k \right)$ and $\left( k+1 , \ 2^{k+1}+n-k-1 \right)$.
It is obvious that  Ineq.~(\ref{ineq:cs}) is facet defining since these points are two independent polytope vertices.
\end{proof}

We now consider the class ${\cal G}(n)$ and study in more details how points $(x,y)$ corresponding to graphs $G$ with $\al(G) = x$ and $F(G) = y$ are situated with respect to the line $y = \Lg{n}{x}$.

\begin{lem} \label{lem:angcoef}
Let $n$ and $\al$ be integers such that $n \ge 7$ and $2 \le \al \le n$, then
\begin{equation} \label{eq:angcoef}
\frac{ \fT(n,\al) }{\al - 1} \le \frac{2^n}{n-1} \cdot
\end{equation}
\end{lem}

\begin{proof}
We consider three cases $\al=n$, $\al=2$ and $3 \le \al \le n-1$. Let $q = n  - \left\lceil \frac{n}{\al} \right\rceil$ be the order of the graph obtained by removing a clique of maximal size in $\tur{n}{\al}$.

\begin{enumerate}[($i$)]
\item Suppose that $\al = n$. In this case, both sides of Ineq.~(\ref{eq:angcoef}) are equal and the result trivially  holds. 
\item Suppose that $\al = 2$. If $n$ is even, $\fT(n,2) = \frac{n^2}{4} + n + 1$ and if $n$ is odd, $\fT(n,2) = \frac{n^2}{4} + n + \frac{3}{4}$. Hence 
$$\fT(n,2) (n-1) \le \left(\frac{n^2}{4} + n + 1\right) (n-1).$$
The latter function is cubic, and thus strictly less than $2^n$ when $n \ge 7$. The result holds in case $\al = 2$.
\item Suppose now that $3 \le \al \le n-1$.  The proof will use an induction on $n$. If $7 \le n \le 10$, Ineq.~(\ref{eq:angcoef}) can be checked by easy computation and we assume that $n \ge 11$. By Lemma~\ref{lem_ind}, we have
$$
\frac{ \fT(n,\al) }{\al - 1} = \frac{ \fT(n-1,\al) }{\al - 1} + \frac{ \fT(q,\al-1) }{\al - 1} \le \frac{ \fT(n-1,\al) }{\al - 1} + \frac{ \fT(q,\al-1) }{\al - 2} \cdot
$$
We can use induction for $\fT(n-1,\al)/(\al - 1)$ because either we fall in case ($i$) when $\al=n-1$, or we stay in case ($iii$). We can also use induction for $\fT(q,\al-1)/(\al - 2)$. Indeed, if $\al-1 = 2$ or $\al -1 = q$, we fall in cases ($ii$) and ($i$), respectively. Otherwise, notice that
\begin{equation} \label{eq:onq}
q = n  - \left\lceil \frac{n}{\al} \right\rceil >  n  - \left\lceil \frac{n}{3} \right\rceil \ge n - \frac{n+2}{3} \cdot
\end{equation}
Hence, $q \ge 7$ when $n \ge 11$, and we fall in case ($iii$). It follows that
$$
\frac{ \fT(n,\al) }{\al - 1} \le \frac{2^{n-1}}{n-2} + \frac{2^q}{q-1} \cdot
$$
As $2^q/(q-1)$ is an increasing function, it is maximum when $q=n-2$. This leads to
$$
\frac{ \fT(n,\al) }{\al - 1} \le \frac{2^{n-1}}{n-2} + \frac{2^{n-2}}{n-3} \le \frac{2^{n-1}}{n-3} + \frac{2^{n-2}}{n-3}  = \frac{3 \cdot 2^n}{4(n-3)} \cdot
$$
To finish the proof, one has to check if
$$
\frac{3}{4(n-3)} \le \frac{1}{n-1} \cdot
$$
This is the case when $n \ge 9$. \qedhere
\end{enumerate}

\end{proof}

\begin{lem} \label{lem:relaxg}
Let $G$ be a graph of order $n \ge 5$ with a stability number $\al$ and a Fibonacci index $\F$, then
\[
\F \le \frac{2^{n} - n -1}{n-1} (\al - 1) + n+1,
\]
with equality if and only if $G \simeq \K{n}$ or $G \simeq \overline{\K{n}}$.
\end{lem}

\begin{proof}
Notice that the right hand side of the inequality in this lemma is equal to $\Lg{n}{\al}$ (see Figures~\ref{fig:polys} and \ref{fig:poly_all}).

The cases $\al = 1$ and $\al = n$ are trivial and correspond to both cases of equality with $G \simeq \K{n}$ and $G \simeq \overline{\K{n}}$, respectively. We  now assume that $2 \le \al \le n-1$ and we prove the strict inequality $\F < \Lg{n}{\al}$.  By Theorem~\ref{thm_genUB}, it suffices to show that $\fT(n,\al) < \Lg{n}{\al}$. The cases $n=5$ and $n=6$ can be easily checked by computation and we suppose that $n \ge 7$.

To achieve this aim, we use the following geometrical argument. For a fixed value of $n$, we consider two lines. The first one is $y= \Lg{n}{x}$ and the second one is the line passing by the points $(1,n+1)$, $(\al,\fT(n,\al))$ corresponding to $\K{n}$ and $\tur{n}{\al}$, respectively. The first line has slope
$$
 \frac{2^{n} - n -1}{n-1},
$$
and the second line has slope
$$
\dfrac{\fT(n,\al) - (n+1)}{\al-1} \cdot
$$
 We now prove that the slope of the second line is strictly less than the slope if the first line. As $\al < n$ and by Lemma~\ref{lem:angcoef},
$$
\dfrac{\fT(n,\al) - (n+1)}{\al-1} < \dfrac{\fT(n,\al)}{\al-1} - \dfrac{n+1}{n-1} \le \dfrac{2^{n-1}}{n-1} - \dfrac{n+1}{n-1}, 
$$
and the result holds.

\end{proof}

We now consider the class ${\cal C}(n)$, and we make the same kind of computations of done in the two previous lemmas, but with respect to the line $y = \Lc{n}{x}$.

\begin{lem} \label{lem:angcoef2}
Let $n$ and $\al$ be integers such that $n \ge 11$ and $2 \le \al \le n-4$, then
\begin{equation} \label{eq:angcoef2}
\frac{ \fT(n,\al) }{\al - 1} \le \frac{2^{n-1}}{n-2} \cdot
\end{equation}
\end{lem}

\begin{proof}
The proof is similar to the proof of Lemma~\ref{lem:angcoef}. We consider three cases $\al=2$, $\al=n-4$ and $3 \le \al \le n-5$.  Let $q = n  - \left\lceil \frac{n}{\al} \right\rceil$.

\begin{enumerate}[($i$)]
\item Suppose that $\al = 2$. Similarly to case ($ii$) in the proof of Lemma~\ref{lem:angcoef}, we have that $\fT(n,2) \le \frac{n^2}{4} + n + 1$. Hence 
$$\fT(n,2) (n-2) \le \left(\frac{n^2}{4} + n + 1\right) (n-2).$$
The latter function is cubic, and thus strictly less than $2^{n-1}$ when $n \ge 9$.
\item Suppose that $\al = n-4$. In this case, and as $n \ge 11$, the Tur\'an graph $\tur{n}{n-4}$ is isomorphic to the disjoint union of four graphs $\K{2}$ and  $n-8$ graphs $\K{1}$. Hence,
$$
\begin{array}{l l}
\dfrac{2^{n-1}}{n-2} - \dfrac{\fT(n,n-4)}{n-5}  & = \dfrac{2^7 \cdot 2^{n-8}}{n-2} - \dfrac{3^4 \cdot 2^{n-8}}{n-5},\\[2.5ex]
& = \dfrac{2^{n-8} \left[ (2^7 - 3^4) n - (5 \cdot 2^7 - 2 \cdot 3^4) \right]}{(n-2)(n-5)},\\[2.5ex]
& = \dfrac{2^{n-8} \left[ 47 n - 478 \right]}{(n-2)(n-5)},
\end{array}
$$
which is positive when $n \ge 11$. Ineq.~(\ref{eq:angcoef2}) holds in this case.
\item Suppose now that $3 \le \al \le n-5$. We use an induction on $n$. If $11 \le n \le 16$, Ineq.~(\ref{eq:angcoef2}) can be checked by  computation and we assume that $n \ge 17$. Similarly to case ($iii$) in the proof of Lemma~\ref{lem:angcoef}, we have
$$
\frac{ \fT(n,\al) }{\al - 1} \le \frac{ \fT(n-1,\al) }{\al - 1} + \frac{ \fT(q,\al-1) }{\al - 2},
$$
and we can use induction for both terms. Indeed for $\fT(n-1,\al)/(\al - 1)$ we fall in case ($ii$) when  $\al=n-5$, or we stay in case ($iii$). For  $\fT(q,\al-1)/(\al - 2)$, since $3 \le \al \le n-5$, we can check that either $\al-1=2$ or $\al-1=q-4$ two cases already treated in ($i$) and ($ii$), or $3 \le \al-1 \le q-5$. In the latter case, we have $q \ge 11$ when $n \ge 17$ by Ineq.~(\ref{eq:onq}). It follows that
$$
\frac{ \fT(n,\al) }{\al - 1} \le \frac{2^{n-2}}{n-3} + \frac{2^{q-1}}{q-2}.
$$
As $2^q/(q-2)$ is increasing, it is maximum when $q=n-2$. This leads to
$$
\frac{ \fT(n,\al) }{\al - 1} \le \frac{2^{n-2}}{n-3} + \frac{2^{n-3}}{n-4} \le \frac{3 \cdot 2^{n-1}}{4(n-4)} \cdot
$$
The proof is completed because
$$
\frac{3}{4(n-4)} \le \frac{1}{n-2}
$$
when $n \ge 10$. \qedhere
\end{enumerate}
\end{proof}

\begin{lem} \label{lem:relaxc}
Let $G$ be a connected graph of order $n \ge 8$ with a stability number $\al$ and a Fibonacci index $\F$, then
\[
\F \le \frac{2^{n-1} - n }{n-2} (\al - 1) + n+1,
\]
with equality if and only if $G \simeq \K{n}$ or $G \simeq \Star{n}$.
\end{lem}

\begin{proof}
Observe that the right hand side of the inequality stated in the lemma is equal to $\Lc{n}{\al}$ (see Figure~\ref{fig:poly_all}).

The cases $\al = 1$ and $\al = n-1$ are trivial and correspond to the two cases of equality. When $2 \le \al \le n-2$, we prove the strict inequality $\F < \Lc{n}{\al}$. The cases $n=8$, $n=9$ and $n=10$ can be checked by computation and we therefore suppose that $n \ge 11$. We consider separately the two cases $2 \le \al \le n-4$ and $n-3 \le \al \le n-2$. 

\begin{enumerate}[($i$)]
\item Let $2 \le \al \le n-4$.  By Theorem~\ref{thm_conUB}, it is enough to show that $\fTC(n,\al) < \Lc{n}{\al}$. We prove a stronger result, that is, $\fT(n,\al) < \Lc{n}{\al}$. The result follows since $\fTC(n,\al) \le \fT(n,\al)$. This situation is well illustrated in Figure~\ref{fig:poly_all} which also indicates that the case $n-3 \le \al \le n-2$ has to be treated separately. 

The argument to prove that $\fT(n,\al) < \Lc{n}{\al}$ is the same as in the proof of Lemma~\ref{lem:relaxg}. We show that the slope of the line $y = \Lc{n}{x}$ is strictly greater than the slope of the line passing by the two points corresponding to $\K{n}$ and $\tur{n}{\al}$.

As $\al \le n-4$, we have
$$
\dfrac{n+1}{\al-1} \ge \dfrac{n+1}{n-5} > \dfrac{n}{n-2}.
$$
This observation and Lemma~\ref{lem:angcoef2} lead to
$$
\dfrac{\fT(n,\al)}{\al-1} - \dfrac{n+1}{\al-1} < \dfrac{\fT(n,\al)}{\al-1} - \dfrac{n}{n-2} \le \dfrac{2^{n-1}}{n-2} - \dfrac{n}{n-2},
$$
and the announced property on the slopes is proved.

\item
Let $n-3 \le \al \le n-2$.  By Theorem~\ref{thm_conUB}, one has to show that $\fTC(n,n-2) < \Lc{n}{n-2}$ and $\fTC(n,n-3) < \Lc{n}{n-3}$. It suffices to prove that $$\fTC(n,n-2) < \Lc{n}{n-3}.$$ Indeed, $\fTC(n,n-3) < \fTC(n,n-2)$ by Corollary~\ref{cor_f2inc} and $\Lc{n}{n-3} < \Lc{n}{n-2}$ because the slope of $y = \Lc{n}{x}$ is strictly positive. As $n \ge 11$, we have $\al \ge \frac{n}{2}$ and we use Corollary~\ref{cor_trees} to compute $\fTC(n,n-2)$. This leads to
$$
\begin{array}{l l}
\Lc{n}{n-3} - \fTC(n,n-2) & =  \dfrac{2^{n-1} - n }{n-2} (n-4) + n+1 - 3 \cdot 2^{n-3} - 2,\\ [2.5ex]
 & =   \dfrac{(n-10) \cdot 2^{n-3} + n + 2}{n-2},
\end{array}
$$
which is strictly positive when $n\ge10$. \qedhere
\end{enumerate}
\end{proof}

We can now give the proof of Theorems~\ref{thm_polyg} and~\ref{thm_polyc}.

\begin{proof}[Proof of Theorems~\ref{thm_polyg} and~\ref{thm_polyc}] 
We begin with the proof for the polytope $\Polyg{n}$. Looking at Lemma~\ref{lem:csineq}, it remains to prove that ($i$) Ineq.~(\ref{eq:otherg}) is facet defining; ($ii$) there are exactly $n$ facet defining inequalities of $\Polyg{n}$. 

The proof ($i$) is straightforward. Indeed, Lemma~\ref{lem:relaxg} ensures that Ineq.~(\ref{eq:otherg}) is valid. Moreover, the points $(1,n+1)$ and $(n,2^n)$ correspond to the graphs $\K{n}$ and $\overline{\K{n}}$, respectively. These points are affinely independent and satisfy Ineq.~(\ref{eq:otherg}) with equality. Therefore Ineq.~(\ref{eq:otherg}) is facet defining.

For ($ii$), it suffices to observe that for any value of $x = 1, 2, \dots, n$, there is exactly one vertex in the polytope: the point which correspond to the complete split graph $\CS{n}{x}$. It follows that $\Polyg{n}$ has exactly $n$ vertices and $n$ facets.

The proof is similar for the polytope $\Polyc{n}$ except that $x < n$. Indeed, Ineq.~(\ref{eq:otherc}) is valid by Lemma~\ref{lem:relaxc}, and the points satisfying Ineq.~(\ref{eq:otherc}) with equality correspond to the graphs $\K{n}$ and~$\Star{n}$.
\end{proof}

\section{Observations} \label{sec_conc}

Tur\'an graphs $\tur{n}{\al}$ have minimum size inside ${\cal G}(n,\al)$ by the Theorem of Tur\'an~\cite{Turan}. Christophe et al.~\cite{GphStableMax} give a tight lower bound for the connected case of this theorem, and Bougard and Joret~\cite{Bougard08} characterized the extremal graphs, which happen to contain the $\turc{n}{\al}$ graphs as a subclass.

By these results and Theorems \ref{thm_genUB} and \ref{thm_conUB}, we can observe the following relations between graphs with minimum size and maximum Fibonacci index. The graphs inside ${\cal G}(n,\al)$ minimizing $m(G)$ are exactly those which maximize $\F(G)$. This is also true for the graphs inside ${\cal C}(n,\al)$, except that there exist other graphs with minimum size than the Tur\'an-connected graphs. 

However, these observations are not a trivial consequence of the fact that $\F(G) < \F(\gme{G}{e})$ where $e$ is any edge of a graph $G$. As indicated in our proofs, the latter property only implies that a graph maximizing $\F(G)$ contains only $\al$-critical edges (and $\al$-safe bridges for the connected case). Our proofs use a deep study of the structure of the extremal graphs to obtain Theorems \ref{thm_genUB} and \ref{thm_conUB}.

We now give additional examples showing that the intuition that more edges imply fewer stable sets is wrong. Pedersen and Vestergaard~\cite{PV06} give the following example. Let $r$ be an integer such that $r\ge 3$, $G_1$ be the Tur\'an graph $\tur{2r}{r}$ and $G_2$ be the star $\Star{2r}$. The graphs  $G_1$ and $G_2$ have the same order but $G_1$ has less edges ($r$) than $G_2$ ($2r-1$). Nevertheless, observe that $F(G_1) = 3^{r} < F(G_2) = 2^{2r-1}+1$. This example does not take into account the stability number since $\al(G_1) = r$ and $\al(G_2) = 2r-1$. 

We propose a similar example of pairs of graphs with the same order and the same stability number (see the graphs $G_3$ and $G_4$ on Figure~\ref{fig_cexa}). These two graphs are inside the class ${\cal G}(6,4)$, however $m(G_3) < m(G_4)$ and $\F(G_3) < \F(G_4)$. Notice that we can get such examples inside ${\cal G}(n,\al)$ with $n$ arbitrarily large, by considering the union of several disjoint copies of $G_3$ and $G_4$.

\begin{figure}[!ht]
\begin{center}
\includegraphics{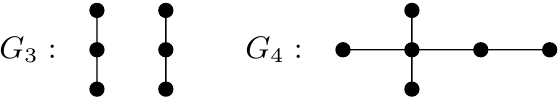}
\caption{Graphs with same order and stability number} \label{fig_cexa}
\end{center}
\end{figure}

These remarks and our results suggest some questions about the relations between the size, the stability number and the Fibonacci index of graphs. What are the lower and upper bounds for the Fibonacci index inside the class ${\cal G}(n,m)$ of graphs order $n$ and size $m$; or inside the class ${\cal G}(n,m,\al)$ of graphs order $n$, size $m$ and stability number $\al$? Are there classes of graphs for which more edges always imply fewer stable sets? We think that these questions deserve to be studied.

\section*{Acknowledgments}

The authors thank Gwena\"el Joret for helpful suggestions.

%

\end{document}